\documentclass[envcountsame,oribibl]{llncs}
\usepackage{amsmath,cite,amssymb,color,trig,tikz,setspace,subfig,appendix,enumerate,xcolor}

\usepackage{graphicx}

\newcommand{\E}{\mathcal{E}}

\newcommand{\G}{\mathcal{G}}
\newcommand{\N}{\mathcal{N}}
\newcommand{\M}{\mathcal{M}}
\newcommand{\D}{\mathcal{D}}
\newcommand{\A}{\mathcal{A}}
\newcommand{\Z}{\mathcal{Z}}
\newcommand{\Q}{\mathcal{Q}}

\newcommand{\Ss}{\mathcal{S}}
\newcommand{\T}{\mathcal{T}}
\newcommand{\B}{\mathcal{B}}
\newcommand{\Cu}{\hat C}

\DeclareMathOperator*{\argmin}{arg\,min}

\definecolor{ceruleanblue}{rgb}{0.16, 0.32, 0.75}

\begin{document}

\newcommand{\localtextbulletone}{\textcolor{gray}{\raisebox{.45ex}{\rule{.6ex}{.6ex}}}}
\renewcommand{\labelitemi}{\localtextbulletone}

% Multiple authors
\title{A Supervisory Control Approach to Dynamic Cyber-Security}
%\title{Determining Optimal Dynamic Cyber-Security Policies: A Supervisory Control Approach}
\author{Mohammad Rasouli\and Erik Miehling\and Demosthenis Teneketzis}
%\\ \email{\{rasouli,miehling,teneket\}@umich.edu}}
\institute{Department of Electrical Engineering and Computer Science\\University of Michigan, Ann Arbor, MI}
\date{}
\maketitle

\begin{abstract}
An analytical approach for a dynamic cyber-security problem that captures progressive attacks to a computer network is presented. We formulate the dynamic security problem from the defender's point of view as a supervisory control problem with imperfect information, modeling the computer network's operation by a discrete event system. We consider a min-max performance criterion and use dynamic programming to determine, within a restricted set of policies, an optimal policy for the defender. We study and interpret the behavior of this optimal policy as we vary certain parameters of the supervisory control problem.
\end{abstract}
% Used "we" too often

 Cyber-Security, Computer Networks, Discrete Event Systems, Finite State Automata, Dynamic Programming

\section{Introduction}
\label{Introduction}

Cyber-security has attracted much attention recently due to its increasing importance in the safety of many modern technological systems. %Cyber-security concerns the protection of systems that contain any computer-based element making 
These systems are ubiquitous in our modern day life, ranging from computer networks, the internet, mobile networks, the power grid, and even implantable medical devices. This ubiquity highlights the essential need for a large research effort in order to strengthen the resiliency of these systems against attacks, intentional and unintentional misuse, and inadvertent failures.

The study of cyber-security problems in the existing literature can be divided into two main categories: \emph{static} and \emph{dynamic}. 

Static problems concern settings where the agents, commonly considered to be an attacker and a defender, receive no new information during the time horizon in which decisions are made. %make their decisions in \emph{one-shot}.
Problems of this type in the security literature can largely be classified under the category of \emph{resource allocation}, where both the defender and attacker make a single decision as to where to allocate their respective resources. The main bodies of work involve infrastructure protection \cite{bier2007,bohme2010,hart2008discrete} and mitigation of malware and virus spread in a network \cite{bloem2007, bloem2009,chen2006, mastroleon2009}. 
%Bier et al. \cite{bier2007} describe a situation where the defender chooses the location of security gates in preparation for an attack.
% and a concept termed \emph{penetration testing} \cite{bohme2010}, where a system is analyzed for potential vulnerabilities.
 Some of the above works consider settings where the agents are strategic \cite{bier2007,hart2008discrete}. The presence of strategic agents results in a game between the attacker and defender. The strategic approaches in the above works are commonly referred to as \emph{allocation games}. The survey by Roy et al. \cite{roy2010survey}, as well as \cite{blotto}, provide useful outlines of some static game models in security.

Dynamic security problems are those that evolve over time, with the defender taking actions while observing some new information from the environment.\footnote{This new information could consist of the attacker's actions, events in nature, or the state of a some underlying system.} The formulation of a security problem as a dynamic problem, instead of a static one, offers numerous advantages. The first advantage is clear; since real-world security problems have an inherently dynamic aspect, dynamic models can more easily capture realistic security settings, compared to static models. Also, most attacks in cyber-security settings are \emph{progressive}, meaning more recent attacks build upon previous attacks (such as denial-of-service attacks, brute-force attacks, and the replication of viruses, malware, and worms, to name a few). This progressive nature is more easily modeled in a dynamic setting than in a static setting. 

The literature within the dynamic setting can be further subdivided into two areas: models based on control theory \cite{rowe2012,khouzani2012maximum,schneider2000,ligatti2005,ligatti2009} and models based on game theory \cite{khouzani2012saddle,yin2010,van2013,roy2010survey}.

The control theory based security models in the literature differ in the ways in which the dynamics are modeled. The work by Khouzani et al. \cite{khouzani2012maximum} studies the problem of a malware attack in a mobile wireless network; the dynamics of the malware spread are modeled using differential equations. %maximum damage malware attack to a mobile wireless network. It models the spread of the malware in the system with differential equations. 
A large part of the literature on control theory based models focuses on problems where the dynamics are modeled by finite state automata. The works of  \cite{ligatti2005,ligatti2009,schneider2000} implement specific control policies (protocols) for security purposes. The work of Schneider \cite{schneider2000} uses a finite state automaton to describe a setting where signals are sent to a computer. Given a set of initial possible states, the signals cause the state of the computer to evolve over time. An entity termed the \emph{observer} monitors the evolution of the system and enforces security in real-time.
% what is the result of the paper?
% The automata enforces security in real-time. 
% what does this mean?
% An observer keeps track of the system states and whenever an illegal transition is going to be made from a set of states, the observer halts the system.
Extensions of Schneider's model are centered around including additional actions for the observer. Ligatti et al. \cite{ligatti2005} extend Schneider's model by introducing a variety of abstract machines which can edit the actions of a program, at run-time, when deviation from a specified control policy is observed. More recent work \cite{ligatti2009} develops a formal framework for analyzing the enforcement of more general policies. 
%{edit automata} which are abstract machines meant to \emph{run-time security policies} (security policies which are )
%Reference  introduces \emph{insertion automaton} and \emph{suppression automaton} to enrich Schneider's automaton. Reference \cite{ligatti2009} discusses the enforcement of non-safety policies.
%the available actions of the observer from binary decisions (continue/stop) to edit the stream of the actions by injecting or deleting some actions. 
Another category of dynamic defense concerns scenarios where the defender selects an \emph{adaptive attack surface}\footnote{For example, changing the network topology.} in order to change the possible attack and defense policies. A notion termed \emph{moving target defense} (a term for dynamic system reconfiguration) is one class of such dynamic defense policies.  The work of Rowe et al. \cite{rowe2012} develops control theoretic mechanisms to determine maneuvers that modify the attack surface in order to mitigate attacks. The work involves first developing algorithms for estimation of the security state of the system, then formalizing a method for determining the cost of a given maneuver. 
The model uses a logical automaton to describe the evolution of the state of the system; however, it does not propose an analytical approach for \emph{determining} an optimal defense policy.

The next set of security models in the literature are based on the theory of dynamic games. %These works consider scenarios where either the attacker or defender \textcolor{red}{(What does a single strategic agent game mean?)}, or both, are strategic; however, unlike in the static setting, the game is played over multiple stages. 
The work in \cite{lye2005} considers a stochastic dynamic game to model the environment of conflict between an attacker and a defender. In this model, the state of the system evolves according to a Markov chain.  This paper has many elements in common with our model; however, it assumes the attacker and defender have perfect observations of the system state. In our paper, we consider the problem from the defender's point of view and assume that the defender has imperfect information about the system state. The work by Khouzani \cite{khouzani2012saddle} studies a zero-sum two-agent (malware agent and a network agent) dynamic game with perfect information. The malware agent is choosing a strategy which trades off malware spread and network damage while the network agent is choosing a counter-measure strategy. The authors illustrate that \emph{saddle-point strategies} exhibit a threshold form. 
%attacks and defense in a dynamic environment where the attackers actions defined to be the rate of infecting safe nodes and killing infected nodes, while defender decides on rate of resurgence. This model does not consider imperfect information for defender regarding attacker actions, and also is actions are taken in contineous time. Their analytical approach takes Nash Equilibrium as solution concept which means it does not consider the sub-games.
%
%
% There are also game-theoretic models in the dynamic setting.
%
The work of Yin et al. \cite{yin2010} (dynamic game version of \cite{bier2007}) studies a Stackelberg game where the defender moves first and commits to a strategy. %This model does not address the progressive attack.
The work addresses how the defender should choose a strategy when it is uncertain whether the attacker will observe the first move. 
Van Dijk et al. \cite{van2013} propose a two player dynamic game, termed \emph{Flipit}, which models a general setting where a defender and an attacker fight (in continuous time) over control of a resource. The results concern the determination of scenarios where there exist dominant strategies for both players.
We refer the reader to Roy et al. \cite{roy2010survey}, and references therein, for a survey on the application of dynamic games to problems in security. 

While models based on game theory have generated positive results in the static setting, there has been little progress in the dynamic setting. We believe this is for two reasons; first, dynamic security has not been fully investigated in a non-strategic context and second, the results in the theory of dynamic games are limited. %In this research, we do not consider the attacker to be strategic and focus on finding the best dynamic defense policy.

In this paper, we develop a (supervisory) control theory approach to a dynamic cyber-security problem and determine the optimal defense policy against progressive attacks. We consider a network of $K$ computers, each of which can be in one of four security states, as seen in Figure \ref{fig:network}. The state of the system is the $K$-tuple of the computer states and evolves in time with both defender and attacker actions. We use a finite state logical automaton to model the dynamics of the system. The defender adjusts to attacks based on the information available.

\begin{figure}[!htbp]
\begin{center}  
\includegraphics [width=1\textwidth,height=\textheight,keepaspectratio]{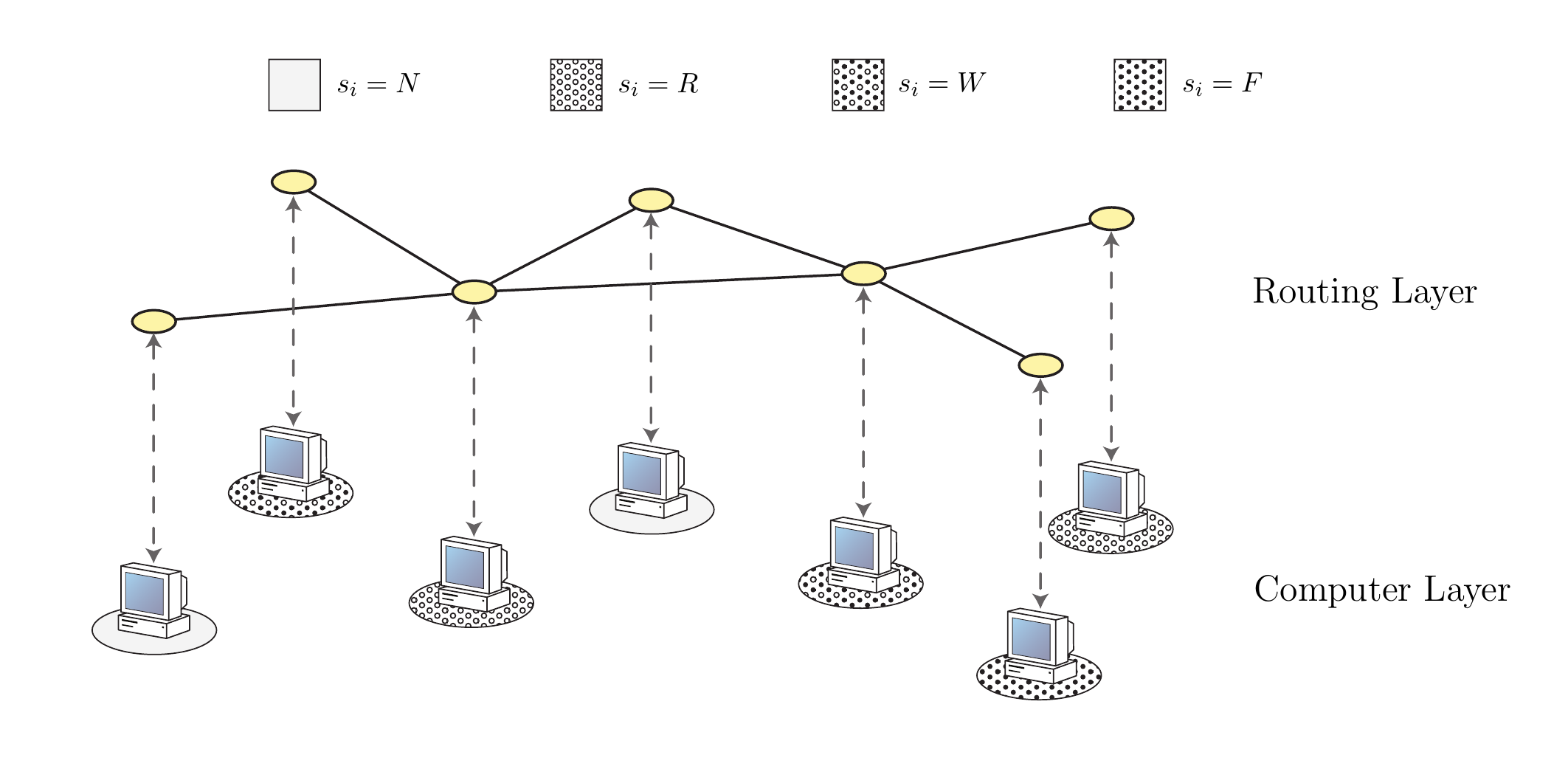}  
\caption{\small \sl An instance of the problem that we consider. Computers are connected through a routing layer. Each computer can be in one of four security states: \emph{normal} (N), \emph{compromised} (R), \emph{fully compromised} (W), or \emph{remote compromised} (F).} \label{fig:network}
\end{center}  
\end{figure}

Our model takes a different approach than the existing papers in the literature. One fundamental difference of our work from the existing literature that make use of automata is the development of an \emph{analytical} framework for \emph{determining optimal defense policies} within a restricted set of policies. Other works involving automata propose methods for \emph{enforcing} a predetermined policy, rather than determining an optimal policy. Also, our control theoretic approach considers imperfect information regarding attacker actions, which we feel is an aspect that is engrained into security problems.

\subsection{Contribution}

The contribution of this paper is the development of a formal model for analyzing a dynamic cyber-security problem from the defender's point of view. Our approach has the following desirable features: (i) It captures the progressive nature of attacks; (ii) It captures the fact that the defender has imperfect knowledge regarding the state of the system; this uncertainty is a result of the fact that all attacks are uncontrollable and most are unobservable, by the defender; (iii) It allows us to quantify the cost incurred at every possible state of the system, as well as the cost due to every possible defender action; (iv) It allows us to quantify the performance of various defender policies and to determine the defender's optimal control policy, within a restricted set of policies, with respect to a min-max performance criterion.
	%\item It allows us to study the sensitivity of the system's performance with respect to the variations in both the cost of states and cost of actions.

\subsection{Organization}

The paper is organized as follows. In Section \ref{sec:model} we discuss our dynamic defense model. This is done by introducing the assumptions on the computer network and corresponding state, as well as the events which drive the evolution of the system state. In Section \ref{sec:formulation}, we model the defender's problem of keeping the computer network as secure as possible while subjected to progressive attacks. We provide a simplified problem formulation that is tractable. In Section \ref{sec:solution}, we determine an optimal control policy for the defender based on dynamic programming. We discuss the nature of the optimal policy in Section \ref{sec:policy}. We offer conclusions and reflections in Section \ref{sec:conclusion}.

\section{The Dynamic Defense Model}
\label{sec:model}

The key features of our model are characterized by assumptions {\bf (A1)} -- {\bf (A6)}. We first describe the assumptions related to the \emph{computer network}, discussed in assumption {\bf (A1)}. In assumption {\bf (A2)} we introduce the notion of the computer network \emph{system state}. Next, in assumptions {\bf (A3)} -- {\bf (A5)}, we discuss the \emph{events} that can occur within the system. We describe how the events cause the system state to evolve, as well as specify which events are controllable and observable by the defender. In {\bf (A6)} we discuss an assumption on the rules of interaction between the attacker and the defender. As mentioned in the introduction, we consider the cyber-security problem from the defender's viewpoint; the model we propose reflects this viewpoint.\\

\noindent {\bf \emph{Assumption 1 - Computer Network}}: \emph{We assume a set of networked computers, $\N=\{1,2,\ldots,K\}$. Each computer, $i\in\N$, can be at security level $z^i\in\M = \left\{N,R,W,F\right\}$ where $\M$ is the set of security states.}
\vspace{0.5em}

Each computer, $i\in\N$, is assumed to have three security boundaries, denoted by $\B = \{B_1,B_2,B_3\}$, representative of a layered structure to its security. These security boundaries partition the set of security states $\M$. Throughout this paper, we assume that the set of security states $\M= \left\{N,R,W,F\right\}$ is defined as follows.\\
\vspace{-0.5em}
\begin{itemize}
\item \emph{Normal} ($z^i=N$): Computer $i$ is in the \emph{normal} state if none of the security boundaries have been passed by the attacker. 

\item \noindent \emph{Compromised} ($z^i=R$): Computer $i$ is \emph{compromised} when security boundary $B_1$ has been passed by the attacker. In this state, the attacker has exploited some vulnerability on the computer
and has managed to obtain user-level access privilege to the computer.%Contents of computer $i$ are readable by the attacker.

\item \noindent \emph{Fully Compromised} ($z^i=W$): Computer $i$ is \emph{fully compromised} when both boundaries $B_1$ and $B_2$ have been passed by the attacker. The attacker has exploited some additional vulnerability on the computer and has managed to obtain root level or execute privilege to the computer. %As a result, contents of computer $i$ are both readable and writeable by the attacker.

\item \noindent \emph{Remote Compromised} ($z^i=F$): Computer $i$ is \emph{remote compromised} when all security boundaries $B_1$, $B_2$, and $B_3$ have been passed by the attacker. The attacker has managed to obtain enough privileges to attack another computer and obtain user-level access privilege on that computer.%Contents of computer $i$ are readable and writeable by the attacker.
% That is, the attacker can use computer $i$ to communicate with, and obtain read and write permissions on, another computer in the network.
\end{itemize}
%The above decomposition is derived from \cite{microsoft}. 
%One interpretation of the boundaries $B_1$, $B_2$, and $B_3$ could be knowing the passwords of the user account, computer admin account, and network admin account, respectively.\\

\noindent {\bf \emph{Assumption 2 - System State}}: \emph{We assume that the computer network operates over an infinite time horizon, $\T=\{0,1,2,\ldots\}$. The state of the computer network, $Z_t$, which evolves with time $t\in\T$, is the combination of the states of all the computers at time $t$. Each state $Z_t$ has a corresponding cost.}
\vspace{0.5em}

The state of the network, denoted $Z_t = (z_t^1,z_t^2,\ldots,z_t^K)\in\Z$, is a $K$-tuple of all of the computer states.\footnote{For example, a three computer network could have a network state of $Z_t'=(N,R,W)$. Notice that state $Z_t'$ is distinct from state $Z_t'' = (R,N,W)$.} The set $\Z$ denotes the set of all possible states, $\Z = \{Z^1,Z^2,\ldots,Z^{|\M|^K}\}=\left\{(N,N,\ldots,N),(N,N,\ldots,R),\ldots,(F,F,\ldots,F)\right\}$, where $|\M|^K$ is the number of system states. 

The cost of the network state $Z_t$ is defined by the costs of the states of the computers. We assign a cost, $c(z_t^i)$, to each computer $i$ depending upon its state $z_t^i\in\M$. This cost is defined as follows
\begin{align}
	c(z_t^i) = \left\{\begin{array}{ll}   
    		c_N &  \text{ if } z_t^i = N\\
		c_R &  \text{ if }z_t^i = R \\		
		c_W &  \text{ if }z_t^i = W \\
		c_F &  \text{ if }z_t^i = F
	\end{array}\right.\label{eq:computercost}
\end{align}
with $0\le c_N<c_R<c_W<c_F<\infty$. The cost of state $Z_t$ is then defined as
\begin{align}
	C_{Z_t} = \sum_{i\in\N}c(z_t^i)\label{eq:statecost}
\end{align}
The state of the network, $Z_t$, evolves in time due to events, which we discuss in the next set of assumptions.\\

\noindent {\bf \emph{Assumption 3 - Events}}: \emph{There is a set of events, $\E = \A\cup\D$, where $\A$ are the attacker's actions and $\D$ are the defender's actions.}%\footnote{A consequence of taking the perspective of the defender is that we do not model the cost of attacker actions.}}
\vspace{0.5em}

We assume that the attacker has access to three types of actions. The set of attacker actions, $\A=\left\{N^a,\{P^i_n\}_{i\in\N,n\in\B},\{H^{ij}\}_{i,j\in\N}\right\}$, is defined as follows.

\begin{itemize} 
\item \emph{$N^a$, null}: The attacker takes no action. The null action does not change the system state and is admissible at any state of a computer.
\item \emph{$P^i_n$, security boundary attack}: Attacking the $n^\text{th}$ security boundary of computer $i$ causes the security state of computer $i$ to transition across the $n^\text{th}$ security boundary. Specifically, $P^i_{B_1}$ causes computer $i$ to transition from normal, $z^i = N$, to compromised, $z^i = R$; $P^i_{B_2}$ from $z^i = R$ to $z^i = W$; and $P^i_{B_3}$ from $z^i=W$ to $z^i = F$.  Actions $P_{B_1}^i$, $P_{B_2}^i$, and $P_{B_3}^i$ are only admissible from states $z^i = R$, $z^i = W$, and $z^i = F$, respectively.

%\item \emph{$A_j$, attack security boundary of computer $j$}: The action $A_j$ is admissible at states $s_i \in\{S,R\}$. Probing the parameter set of computer $j$ will cause the state of the computer $j$ to transition from safe, $s_j = S$, to readable, $s_j = R$, or from readable, $s_j = R$, to writeable, $s_j = W$, depending on the current state of computer $j$.
%\item \emph{$I_j$, infiltrate computer $j$}: The action $I_j$ is admissible at state $s_i = W$. Infiltrating computer $j$ will bring the state from writeable, $s_j = W$, to fully controllable, $s_j = C$.
\item \emph{$H^{ij}$, network attack}:  Using a computer $i$ in state $z^i = F$ to attack any other normal or compromised computer $j$ in the network that is in state $z^j = \{N,R\}$ to bring computer $j$ to state $z^j = W$. The action $H^{ij}$ is admissible at state $z^i = F$ for $z^j\in\{N,R,W\}$.%$j$, $s_j = S$ or $s_j = Q$, in the network.
\end{itemize}
We assume that the defender knows the set $\A$ as well as the resulting state transitions due to each action in $\A$.
% the defender only needs to know where the observable events in A take the system, 

The defender has access to three types of costly actions. These actions are admissible at any computer state. The set of defender actions, denoted by $\D=\left\{N^d,\{E_i\}_{i\in\N},\{R_i\}_{i\in\N}\right\}$, is defined as follows.
\begin{itemize}
\item \emph{$N^d$, null}: The defender takes no action. The null action does not change the system state.
\item \emph{$E_i$, sense computer $i$}: The \emph{sense} action, $E_i$, reveals the state of computer $i$ to the defender. The sense action does not change the system state.
\item \emph{$R_i$, re-image computer $i$}: The \emph{re-image} action, $R_i$, brings computer $i$ back to the normal state from any state that it is currently in. For example, $R_3$ applied to state $(N,R,F)$ results in $(N,R,N)$.
%The defender knows the computer is in the "Safe" state. % I don't think this statement is necessary 
\end{itemize}
The costs of the actions in $\D$ are defined by $\Cu(N^d)$, $\Cu(E_i)$, $\Cu(R_i)$, where $0\le\Cu(N^d)<\Cu(E_i)<\Cu(R_i)<\infty$ for all $i\in\N$.\\

\noindent {\bf \emph{Assumption 4 - Defender's Controllability of Events}}: \emph{The actions in $\A$ are uncontrollable whereas the actions in $\D$ are controllable.}
\vspace{0.5em}

Since the problem is viewed from the perspective of the defender, all actions in $\D$ are controllable. For the same reason, the defender is unable to control any of the attacker's actions $\A$.\\

\noindent {\bf \emph{Assumption 5 - Defender's Observability of Events}}: \emph{All actions in $\D$ and some actions in $\A$ are assumed to be observable.}
\vspace{0.5em}

Again, due to taking the defender's viewpoint, all actions in $\D$ are observable. Although we assume that the defender knows the set $\A$, we assume that it cannot observe $N^a$ or any $P^i_n$ actions; it can only observe actions of the type $H^{ij}$. One justification for this is that the the network attack $H^{ij}$ involves passing sensitive information of computer $j$ through the routing layer of the system to computer $i$.\footnote{This sensitive information could be the login credentials of computer $j$.} We assume that the routing layer is able to detect the transfer of sensitive data through the network, and thus the defender is aware when an action of the form $H^{ij}$ occurs.\\

\noindent {\bf \emph{Assumption 6 - Defender's Decision Epochs}}: \emph{The defender acts at regular, discrete time intervals. At these time intervals, the defender takes only one action in $\D$. The attacker takes one action in $\A$ between each defender action.}
\vspace{0.5em}

We require that the defender should consider taking a single action in $\D$ at regular time instances. We assume that between any two such instances, the attacker can only take one 
% this discussion is backwards, the attacker acts, then the defender responds.
action in $\A$. This order of events is illustrated in Figure \ref{fig:timeline} for a given time $t=\tau$. We introduce \emph{intermediate states}, denoted by $\tilde \Z = (\tilde Z^1,\tilde Z^2,\ldots,\tilde Z^{|\M|^K})$, which represent the system states at which events from $\A$ are admissible (that is, the states in which the attacker takes an action). The \emph{system states}, denoted by $\Z = (Z^1,Z^2,\ldots,Z^{|\M|^K})$, are the states at which actions from $\D$ are admissible. 
%We use $\Z_t$ to denote the state of the system at discrete times, $t\in\T$ before the defender makes an action. We also use $\tilde \Z$ to denote the \emph{intermediate state} of the system when the defender has taken an action and the network state is updated, but the attacker has not yet made any action.
\begin{figure}[!htbp]
\begin{center}  
\includegraphics [width=0.75\textwidth,height=\textheight,keepaspectratio] {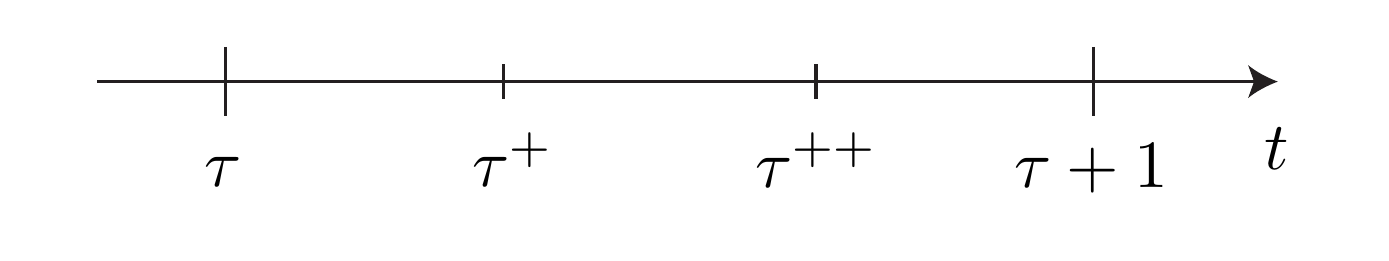}  
\caption{\small \sl Order of events for a given time-step. At time $t=\tau$, the cost of the current state $C_{Z_\tau}$ is realized. At $\tau^+$,  the defender takes an action in $\D$ (the cost of which is realized immediately). The resulting system state due to the defender's action is denoted by $\tilde Z_{\tau^+}\in\tilde\Z$. At $\tau^{++}$ the attacker takes an action in $\A$. At $\tau+1$, the resulting system state is denoted by $Z_{\tau+1}\in\Z$.} 
%\caption{\small \sl Order of events for a given time-step. At time $t=\tau$, the cost of the current state $C_{Z_\tau}$ is realized. At $\tau^+$,  the defender takes an action in $\D$, the cost of which is realized immediately. The state of the network is updated to $\tilde \Z_{\tau^+}$. At $\tau^{++}$ the attacker takes an action in $\A$. At $\tau+1$, the state of the system is updated to $Z_{\tau+1}$.}
\label{fig:timeline}
\end{center}  
\end{figure}
%\vspace{-3em}

Assumption {\bf (A6)} is, in our opinion, reasonable within the security context. Since time has value in security problems,\footnote{A computer that is compromised by the attacker for two time steps is more costly to the defender than a computer that is compromised for one time step.} the defender should take actions at regular time intervals (note that at these instances the defender may choose $N^d$, that is, choose to do nothing). In general, a finite number of events in $\A$ may occur between any two successive defender actions; however, to reduce the dimensionality of the problem, we assume that only one event in $\A$ can occur.\\

One important implication of assumption {\bf (A6)} is related to the defender's observability of events in $\A$. By {\bf (A6)}, the defender is aware \emph{when} an event in $\A$ occurs. Since the event $H^{ij}$ is observable, if the defender does not observe $H^{ij}$ when an event in $\A$ is known to occur, then it knows that one of the unobservable events, $N^a$ or one of $\{P_n^i\}_{i\in\N,n\in\B}$, has occurred. To incorporate this fact into the defender's knowledge about the system's evolution, we group the above mentioned unobservable events into one event, denoted $X=\left\{N^a,\{P^i_n\}_{i\in\N,n\in\B}\right\}$. This philosophy is used in constructing the system automaton from the defender's point of view, as well as in defining the defender's information state (discussed in Section \ref{sec:formulation}). As a result of the above grouping, the set of events $\A'= \left\{X, \{H^{ij}\}_{i,j\in\N}\right\}$ is observable by the defender. Notice, however, that by performing this grouping, we have introduced non-determinism into the system; that is, the event $X$ can take the system to many possible system states. All unobservable events in the problem have been eliminated due to Assumption {\bf (A6)} and the grouping of unobservable events in $\A$.

As a result of assumptions {\bf (A1)} -- {\bf (A6)}, the evolution of the system state, $Z_t$, from the defender's viewpoint, can be modeled by a discrete event system represented by a finite state automaton, which we term the \emph{system automaton}. 
%the evolution of which is based on the events from $\E$ that occur. %interactions between the attacker and the defender. 
%In order to construct the system automaton, we require that there is an entity in the system, termed the \emph{observer}, who has knowledge of the possible network states as well as the attacker and defender action sets. 
%The defender is able to construct the possible transitions between system states by considering the possible events in $\E'$ that can occur. 
Due to assumption {\bf (A6)}, we duplicate the system states by forming the set of intermediate states, denoted by $\tilde \Z = (\tilde Z^1,\tilde Z^2,\ldots,\tilde Z^{|\M|^K})$. The set of intermediate states represents the states at which an event from $\A$ can occur. The set of system states, denoted by $\Z$, are the states at which the defender takes an action $d\in\D$. The resulting automaton has $2|\M|^K$ states. The set of events that can occur is described by the set $\E' = \A'\cup\D$; the transitions due to these events follow the rules discussed in assumption {\bf (A3)}. The system automaton takes the form of a bipartite graph, as seen in Figure \ref{fig:bipartite}.
% Should we add a third type of agent (the observer) under the assumptions
\begin{figure}[!htbp]
\begin{center}  
\includegraphics [width=1\textwidth,height=\textheight,keepaspectratio] {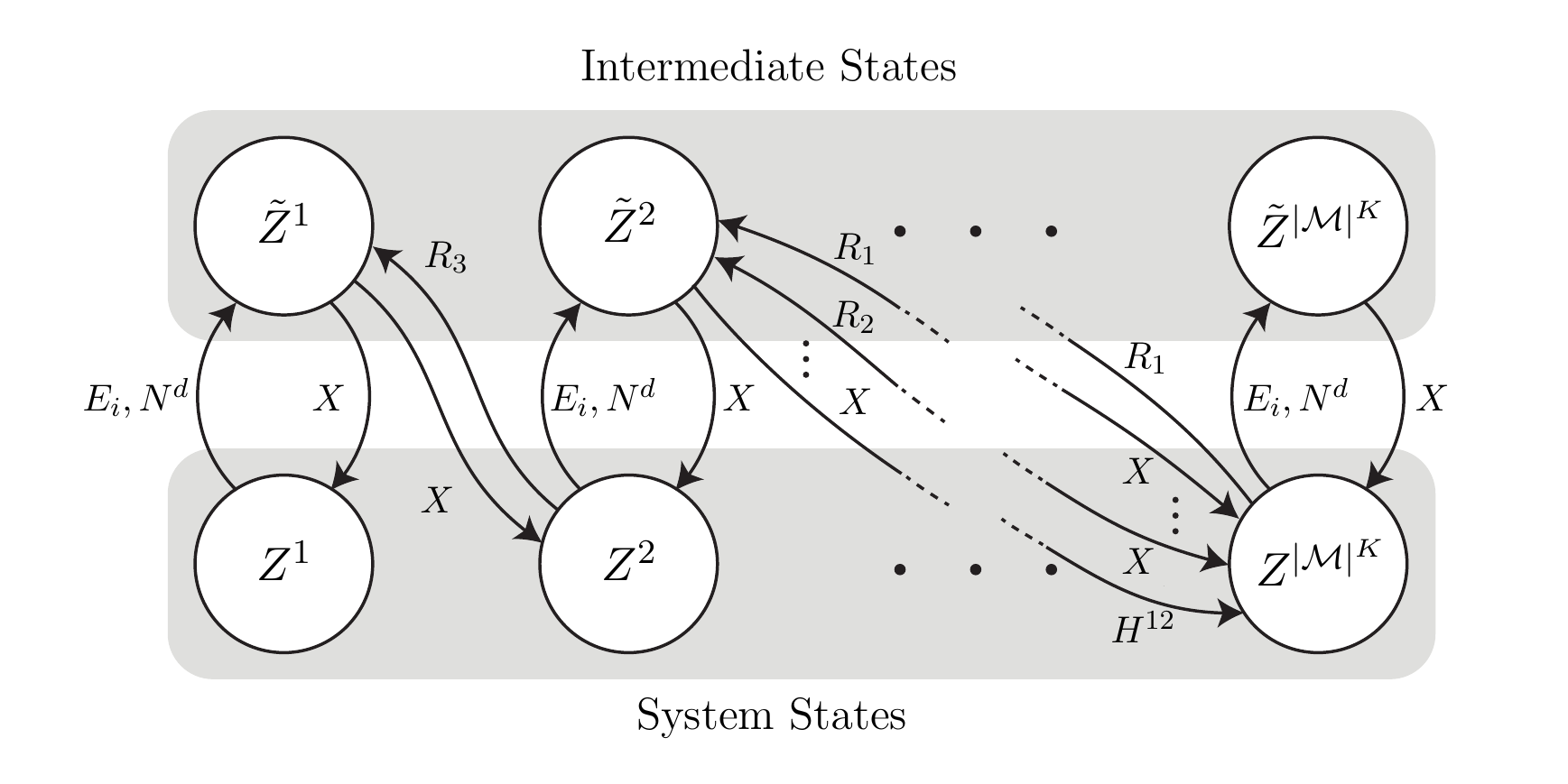}  
\caption{\small \sl The system automaton represented as a bipartite graph of intermediate states, $\tilde \Z = (\tilde Z^1,\tilde Z^2,\ldots,\tilde Z^{|\M|^K})$, and system states, $\Z = (Z^1,Z^2,\ldots,Z^{|\M|^K})$, with events $\E'=\A'\cup\D$. Notice the non-determinism of the event $X\in\A'$.} \label{fig:bipartite}
\end{center}  
\end{figure}
\noindent Notice that, like the null action, the sense actions, $E_i$, for all $i\in\N$, do not change the underlying system state. The purpose of sense is to update the defender's information state, which will be defined and explained in the following section.

\section{The Defender's Problem}
\label{sec:formulation}

We now formulate the defender's problem -- protecting the computer network. The defender must decide which costly action to take, at each time step, in order to keep the system as secure as possible given that it has imperfect knowledge of the network's state. %As a result, the defender's problem is a supervisory control problem with imperfect information. The notion of an information state \cite{kumar1986stochastic} is a key concept in supervisory (and general) control problems with imperfect information. 

%In Section \ref{ssec:infostate}, we specify an information state (finite dimensional) for the defender. In Section \ref{ssec:minmax}, we define the problem's performance criterion and the defender's optimization problem.

\subsection{The Defender's Optimization Problem}
\label{ssec:minmax}

Let $g:=\{g_t, t\in\T\}$, denote a control policy of the defender, where
\begin{align}
	g_t:\D^{t-1}\times{\A'}^{t-1}\to \D\label{eq:g},
\end{align}
and $\D^{t-1}$ and ${\A'}^{t-1}$ denote the space of the defender's actions and observations up to $t-1$, respectively. Let $\G := \{g\,|\,g_t:\D^{t-1}\times{\A'}^{t-1}\to\D\text{ for all }t\in\T\}$ denote the space of admissible control policies for the defender.

The defender's optimization problem is
\begin{align}
\min_{{g\in\G}} \max_{\{{Z_t^g\in\Z},t\in\T\}}& \left\{\sum_{t\in\T} \beta^t  \bigg[C_{Z_t^g}+\hat{C}\big(d_t\big)\bigg]\right\}\label{eq:defprob}\tag{$P_D$}\\
\nonumber  \text{subject to }&\,\,\, \text{Assumptions  {\bf (A1)} -- {\bf (A6)}}
%\nonumber&\,\,\, d_t = g_t(S_t)\\
%\nonumber&\,\,\, Z_t\in\Ss_t\\
%\nonumber&\,\,\, S_{t+1} = f(S_t, d_t,a_t')
\end{align}
where $\{{Z_t^g\in\Z},t\in\T\}$ denotes a sequence of states generated by control policy $g$ and $d_t$ is the defender's action at $t$ generated according to Equation (\ref{eq:g}). Problem (\ref{eq:defprob}) is a supervisory control problem with imperfect observations.

\subsection{Discussion of Problem (\ref{eq:defprob})}
\label{ssec:discpd}

The notion of an information state \cite{kumar1986stochastic} is a key concept in supervisory (and general) control problems with imperfect information. Because of the nature of the performance criterion and the fact that the defender's information is imperfect, an appropriate information state for the defender at time $t$ is $\sigma(\D^{t-1},{\A'}^{t-1})$, the $\sigma$-field generated by the defender's actions and observations, respectively, up to $t-1$. Using such an information state, one can, in principle, write the dynamic program for Problem (\ref{eq:defprob}). Such a dynamic program is computationally intractable. %For this reason, we formulate another problem, called $(P_D')$, where, by restricting attention to a set of defense policies that have a specific structure, we can obtain a computationally tractable solution. %and its solution is computationally tractable.
For this reason, we formulate another problem, called $(P_D')$, where we restrict attention to a set of defense policies that have a specific structure; in this problem we can obtain a computationally tractable solution.

%the solution of which leads to a performance that is inferior to that of Problem (\ref{eq:defprob}), but is computationally tractable.

\subsection{Specification of Problem $(P_D')$}
\label{ssec:infostate}

We define the \emph{defender's observer} as follows. The defender's observer is built using the defender's observable events, $\A'$, and its actions, $\D$.
%The defender's information state captures this uncertainty and is specified through the determination of the \emph{defender's observer}.
%Given that the states of the network and admissible events from each state were defined including their controllability and observability (in Assumptions 1-5), we can construct an information state for the defender.
%First, we introduce the set of possible states of the network at time $t$ from the defender's perspective as a 
The observer's state at time $t$, denoted by $S_t\subseteq \Z$, consists of the possible states that the network can be in at time $t$ from the defender's perspective. We denote by $\Ss$ the space to which $S_t$ belongs, for any $t\in\T$.

%sufficient statistic for the observer's state of the defender denoted by $S_t$, where $S_t\subseteq \Z$ for every $t$. 

The evolution of the observer's state is described by the function $f:\Ss\times\D \times \A' \to\Ss$.  The observer's state $S_t$ follows the update
\begin{align*}
S_{t+1}= f(S_t,d_t,a_t')	
\end{align*}
where $d_t\in\D$ is the realization of the defender's action and its effect at time $t^+$, and $a_t'\in \A'$ is the realization of the defender's observation at $t^{++}$. The precise form of the function $f$ is determined by the dynamic defense model of Section 2. Thus, the dynamics of the defender's observer are described by a finite state automaton with state space $\Ss$ and transitions that obey the dynamics defined by the function $f(S_t, d_t, a_t')$.

Using the defender's observer we formulate Problem $(P_D')$ as follows.
\begin{align}
\min_{{g\in\G'}} \max_{{Z_t^g\in\Z},t\in\T}& \left\{\sum_{t\in\T} \beta^t  \bigg[C_{Z_t^g}+\hat{C}\big(d_t\big)\bigg]\right\}\label{eq:defprobprime}\tag{$P_D'$}\\
\nonumber  \text{subject to }&\,\,\, \text{Assumptions  {\bf (A1)} -- {\bf (A6)}},\\
\nonumber&\,\,\, d_t = g_t(S_t),\,t\in\T,\\
\nonumber&\,\,\, Z^g_t\in S_t,\,t\in\T,\\
\nonumber&\,\,\, S_{t+1} = f(S_t, d_t,a_t'),\,t\in\T.
\end{align}
where $\G' := \{g\,|\,g:=\{g_t,t\in\T\}, g_t:\Ss\to\D\text{ for all }t\in\T\}$.

\section{Dynamic Programming Solution for the Defender's Problem}
\label{sec:solution}

\subsection{The Dynamic Program}

We solve Problem (\ref{eq:defprobprime}) using dynamic programming. The dynamic program corresponding to Problem (\ref{eq:defprobprime}) is
\begin{align}
\label{eq:dynprog}
V(S) =  \min_{d\in \D} \max_{Z \in S}\bigg[C_{Z}+\hat{C}(d)+ \max_{S'\in \Q(S, d,Z)} \beta V(S')\bigg].
\end{align}
for every $S\in\Ss$ (see \cite{kumar1986stochastic,bertsekas1995dynamic}), where $\Q(S, d,Z)$ is the set of observer states that can be reached by $S$ when the defender's action is $d$ and the true system state in $\Ss$ is $Z$. The set $\Q(S, d,Z)$ is determined as follows. If at time $t$ the observer's state is $S$ and the defender takes action $d$ then, before the effect of $d$ at time $t^+$ and the observation at time $t^{++}$ are realized, there will be several potential candidate observer states at $t+1$. Only a subset of these possible observer states can occur when the true state of the system at time $t$ is $Z\in S$. This subset is $\Q(S, d,Z)$. We illustrate the form of the set $\Q(S, d,Z)$ by the following example.

\noindent {\bf Example 1.} Assume a network of three computers and a current observer state of
\begin{align*}
	S_t = \{ (F,N,N),\,(F, N, R),\,(F,R,N) \}.
\end{align*}
If the defender takes action $E_2$ then, before the effect of $E_2$ and the observation $H^{1,2}$ at $t^{++}$ are realized, the possible observer states $S_{t+1}$ are% (since there are two possible states in $S_t$ for computer $i=2$)
\begin{align*}	
	\big\{\{(F,W,N),\,(F,W,R)\},\,\{(F,W,N)\}\big\}.
\end{align*}
If the true system state is $Z_t = (F,N,R)$ then 
\begin{align*}
	\Q(S_t,E_2,Z_t) = \{(F,W,N),\,(F,W,R)\}.
\end{align*}\hfill $\triangle$
%\textcolor{ceruleanblue}{If we hypothesize that the true system state is $Z_t = (F,N,R)$ then the set of possible observer states for time $t+1$ will be
%\begin{align*}
%	S_{t+1}(Z_t) = \{(F,W,N),\,(F,W,R)\}.
%\end{align*}}
%Notice that $S_{t+1}$, under both $f$ and $\hat f$, belongs to a set of possible observer states.

\subsection{Solution of the Dynamic Program}

% should u be u_t?
% what is the correct notation here? R(S_t,u) or R(S,u), what does S represent?
We obtain the solution of the dynamic program, Equation (\ref{eq:dynprog}), via \emph{value iteration} \cite{kumar1986stochastic,bertsekas1995dynamic}. For that matter, we define the operator $T$ by
\begin{align}
	TV(S) :=\min_{d\in \D} \max_{Z \in S}\bigg[C_{Z}+\hat{C}(d)+ \max_{S'\in \Q(S, d,Z)} \beta V(S')\bigg].\label{eq:mapT}
\end{align}
We prove the following result.
%To find the optimal policy of the defender from each of the states in its observer state, we need to do dynamic programming.
\begin{theorem}
\label{contraction}
The operator $T$, defined by Equation (\ref{eq:mapT}), is a contraction map.
\end{theorem}
% We need to define this mapping, it is ambiguous right now
\begin{proof}
We use Blackwell's sufficiency theorem (Theorem 5, \cite{blackwell1965}) to show that $T$ is a contraction mapping. We show:
\begin{enumerate}[i)]
\item \emph{Bounded value functions}: First, note that $|\Ss|,|\D|<\infty$, and that we have bounded costs, $C_{Z}\le M_1<\infty$, $\forall\, S\in\Ss$; $\hat{C}(d)\le M_2<\infty$, $\forall\,d\in\D$. Starting from any bounded value function, $V(S)\le M_3<\infty$ with $M_3 > \frac{M_1+M_2}{1-\beta}$ we have
\begin{align*}
TV(S)&= \min_{d\in \D} \max_{Z \in S}\bigg[C_{Z}+\hat{C}(d)+ \max_{S'\in \Q(S,d,Z)} \beta V(S')\bigg]\\ &\le M_1+M_2+\beta M_3 < M_3 < \infty
\end{align*}
for all $S\in\Ss$.\\
\item \emph{Monotonicity}: Assume $V_2(S) \ge V_1(S)$ $\forall\, S\in\Ss$. Then, for all $S\in\Ss, Z\in S$ and $d\in D$,
\begin{align*}
C_{Z}+\hat{C}(d)+ \max_{S'\in \Q(S,d,Z)} \beta V_2(S') \ge  C_{Z}+\hat{C}(d)+ \max_{S'\in \Q(S,d,Z)} \beta V_1(S')
\end{align*}
Therefore, for all $S\in\Ss$ and $d\in D$
\begin{align*}
&  \max_{Z \in S}\bigg[C_{Z}+\hat{C}(d)+ \max_{S'\in \Q(S,d,Z)} \beta V_2(S')\bigg]  \ge \nonumber\\
&  \max_{Z \in S}\bigg[C_{Z}+\hat{C}(d)+ \max_{S'\in \Q(S,d,Z)} \beta V_1(S')\bigg]
\end{align*}
Hence,
\begin{align*}
TV_2(S) &= \min_{d\in \D} \max_{Z \in S}\bigg[C_{Z}+\hat{C}(d)+ \max_{S'\in \Q(S,d,Z)} \beta V_2(S')\bigg]\\
&\ge  \min_{d\in \D} \max_{Z \in S}\bigg[C_{Z}+\hat{C}(d)+ \max_{S'\in \Q(S,d,Z)} \beta V_1(S')\bigg]=TV_1(S).
\end{align*}
\item \emph{Discounting}: Assume $V_2(S)=V_1(S)+a$. Then, for all $S\in\Ss$
\begin{align*}
TV_2(S) &=  \min_{d\in \D} \max_{Z \in S}\bigg[C_{Z}+\hat{C}(d)+ \max_{S'\in \Q(S, d,Z)} \beta (V_1(S')+a)\bigg]\\
&= \min_{d\in \D} \max_{Z \in S}\bigg[C_{Z}+\hat{C}(d)+ \max_{S'\in \Q(S, d,Z)} \beta V_1(S')\bigg]+\beta a\\
& =TV_1(S)+\beta a.
\end{align*}
\end{enumerate}
By Blackwell's sufficiency theorem, the operator $T$ is a contraction mapping.
\hfill$\square$
\end{proof}
Since $T$ is a contraction mapping, we can use value iteration to obtain the solution to Equation (\ref{eq:dynprog}), which we term the stationary value function, $V^*(S)$. %By Theorem \ref{contraction}, the value iteration procedure converges to the stationary value function $V^*(S)$, and thus we can write the dynamic programming optimality equation as
%\begin{align*}
%V^*(S) = C(S)+\min_{d\in\D} \left\{\hat{C}(d)+\max_{S^{'}\in \Rs(S, d)} \beta V^*(S') \right\},\,S\in\Ss.
%\end{align*}
From the stationary value function, we can obtain an optimal policy, $g^*$, as follows
\begin{align*}
	g^*(S) = \argmin_{d\in\D} \max_{Z \in S}\bigg[C_{Z}+\hat{C}(d)+ \max_{S'\in \Q(S, d,Z)} \beta V(S')\bigg]
\end{align*}
%To use value iteration we use code DefenderValueIteration.m which takes the .fsm defender observer state and calculates the optimal policy.
%Note that value iteration is a standard approach and starting from initial point $0$, at the $n$th iteration it gives us the the best policy for a finite horizon of $n$.
The optimal policy, $g^*(S)$, is not always unique. That is, for a given observer state $S\in\Ss$, there could be multiple $d\in\D$ which achieve the same minimum value of $ \min_{d\in \D}\max_{Z \in S}\bigg[C_{Z}+\hat{C}(d)+ \max_{S'\in \Q(S, d,Z)} \beta V(S')\bigg]$. We denote by $\D^*(S)$ the set of optimal actions for a given observer state $S$. In the event that $\D^*(S)$ is not a singleton for a given state $S$, we choose a single action $d^*(S)\in\D^*(S)$ based on a quantity we define as the \emph{confidentiality threat}. The confidentiality threat is a measure of the degree to which computer $i$ is presumed (by the defender) to be compromised and is defined as follows
\begin{align*}
	\tilde T_i = \sum_{Z\in S}c(z^i),\,\,\,i\in\N
\end{align*}
where $c(z^i)$, $z^i\in\M$, is the cost of the state, as defined in Equation (\ref{eq:computercost}), of the $i^{\text{th}}$ computer in the candidate system state $Z\in S$. Summing over all candidate system states in the observer state $S$ for a given computer $i$, we obtain the confidentiality threat $\tilde T_i$. Next, we compare the confidentiality threat of each computer and choose the action $d^*(S)\in\D^*(S)$ that corresponds to the highest confidentiality threat. In the case of equal confidentiality threats (which arise when the observer state is symmetric), we choose the action in $\D^*(S)$ corresponding to the computer with the lower index $i\in\N$.\footnote{This choice is arbitrary; we could randomize the choice as well.}

\section{Optimal Defender's Policy}
\label{sec:policy}

We now discuss the characteristics of the optimal policy for Problem (\ref{eq:defprobprime}), henceforth referred to as the \emph{optimal policy}. %In Section \ref{ssec:analytical}, we discuss some analytical results related to sensitivity analysis of the policy. In Section \ref{ssec:numerical}
We illustrate sensitivity analysis via numerical results for both a two computer and a three computer network. We also discuss some qualitative observations of the optimal policy.

First we note that determining the set of observer states and its associated dynamics is not a trivial computational task, even for moderately sized networks. Our calculations show for the case of a two computer network, the defender's observer automaton consists of $87$ states and $1207$ transitions. Extending the system to a three computer network results in $1423$ states with $65602$ transitions. To automate the procedure, we have developed a collection of programs which makes use of the UMDES-LIB software library \cite{umdes1}. The specific procedure is discussed in Appendix A. %\ref{appendix:1}.

The sensitivity analysis studies how the cost of re-imaging affects the optimal policy. For both the two computer and three computer networks, we increase the re-image cost, $\hat{C}(R_i) = r, \,\forall i\in\N$, and observe how the optimal policy behaves. Since the number of observer states in the two computer network, denoted $|\Ss_2|$, is modest, $|\Ss_2| = 87$, we are able to plot the behavior for each observer state $S\in\Ss_2$, as seen in Figure \ref{fig:reimage}(a).\footnote{The ``ordering" of these states is arbitrary.} In the three computer network, the size of observer state space, $|\Ss_3|=1423$, is much larger than that of the two computer network. As a result, we plot the percentage of observer states that have the optimal action $d$, for all $d\in\D$, and analyze how the percentage changes as we increase $r$, as seen in Figure \ref{fig:reimage}(b).
%We now present numerical results for two test systems; a two computer and a three computer network.
\begin{figure}[!htbp]
\subfloat[\emph{Two computer network.} Optimal actions for each observer state as a function of increasing re-image cost, $r=3,\ldots,30$.]{\includegraphics [width=0.45\textwidth,height=\textheight,keepaspectratio] {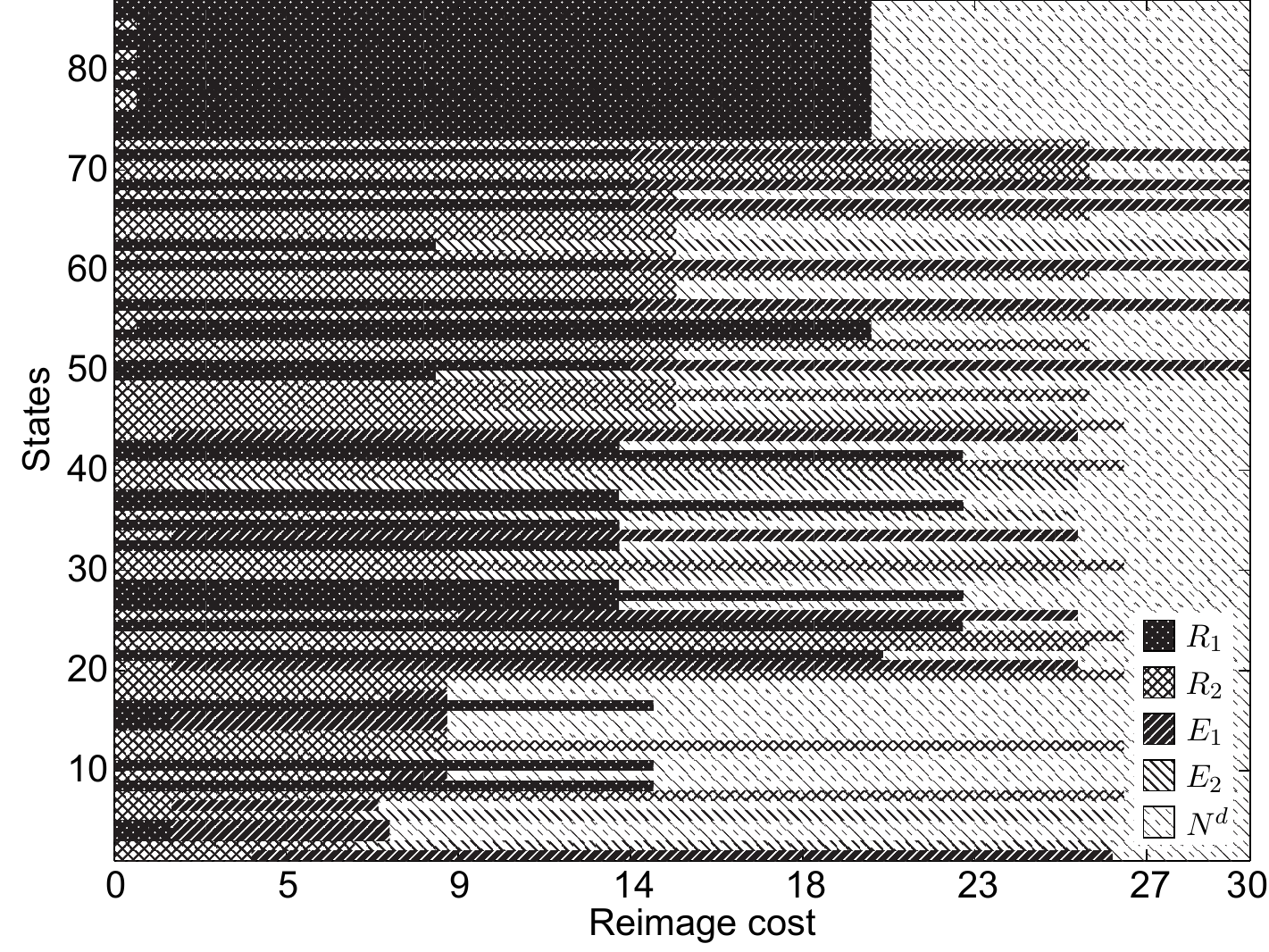}}\hspace{0.5cm}
\subfloat[\emph{Three computer network.} Percentage of observer states that have optimal actions $d\in\D$ as a function of increasing re-image cost, $r=0.2,\ldots,60$.]{\includegraphics [width=0.46\textwidth,height=\textheight,keepaspectratio] {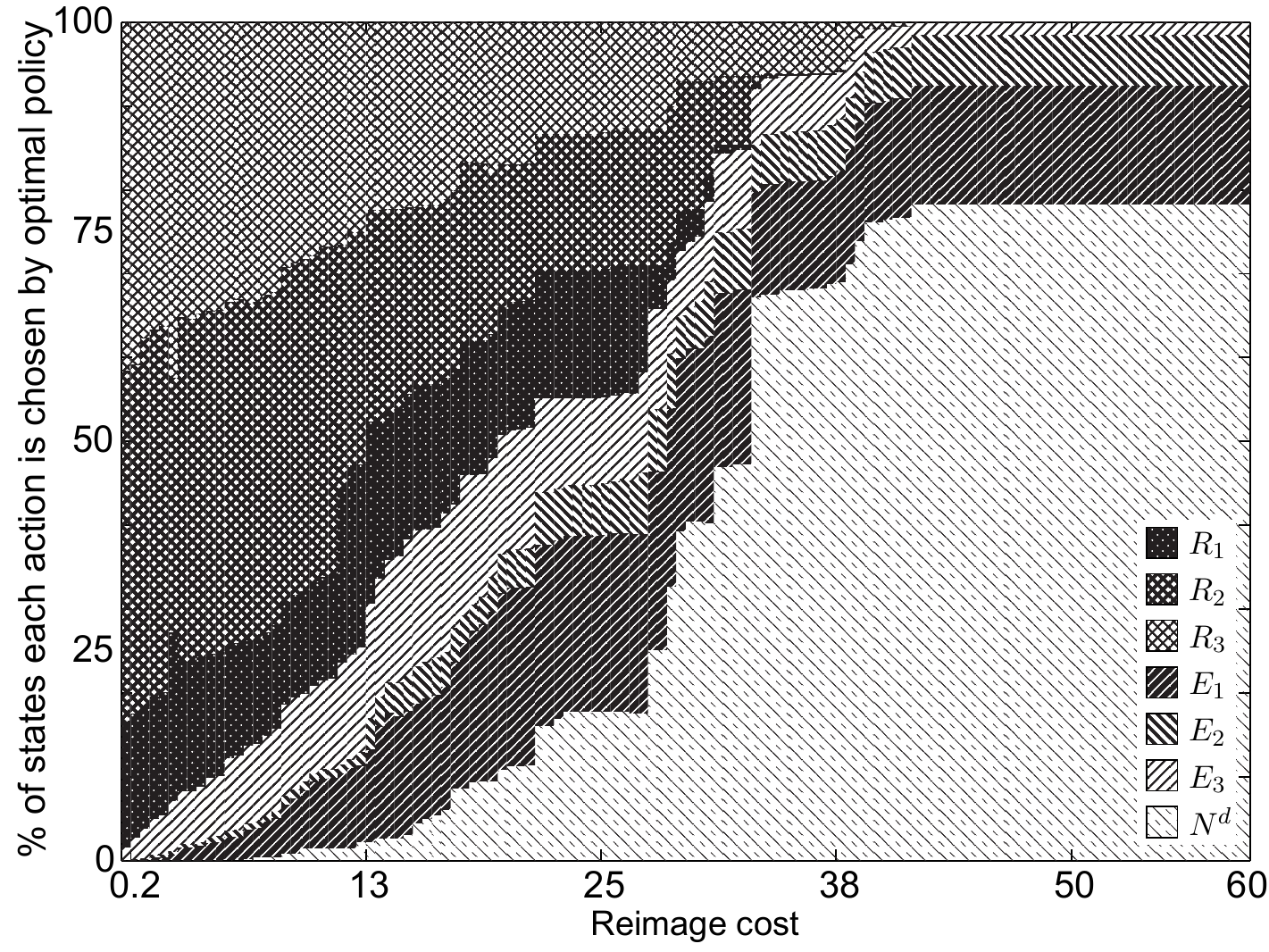}}
\caption{\small \sl Sensitivity analysis for varying re-image cost $r$, where $r = \Cu(R_i)$ for all $i\in\N$. Other parameters are $\Cu(N_d) = 0$, $\Cu(E_i) = 0.1$ $\forall\,i\in\N$, $c_N = 0$, $c_R = 1$, $c_W = 2$, $c_F = 8$, and $\beta = 0.9$.}
\label{fig:reimage}
\end{figure}

The behavior of the optimal policy due to increasing re-image costs, $r$, is intuitive. As $r$ increases, the optimal policy exhibits a threshold form,\footnote{In the simulations that we have performed.} switching from specifying more expensive actions to less expensive actions. For very low re-image costs, the optimal policy specifies $R_i$ in the majority of the observer states. As $r$ increases, observer states for which $R_i$ was optimal, switch to either sense, $E_i$, or null, $N^d$. Once the optimal action is null, it remains null for all higher values of $r$. For the observer states where the action switched to sense, a further increase in $r$ may result in a switch to null; however, there exist some observer states where the optimal action is sense for all higher values of $r$. This threshold behavior is clearly depicted in Figure \ref{fig:reimage}(a).

As a result of the aforementioned threshold behavior, for high enough values of $r$, the optimal policy eventually specifies $N^d$ or $E_i$ for all states $S\in\Ss$. The argument to see why there is no re-image action for high values of $r$ is straightforward; at these values of $r$ the cost of re-imaging is prohibitively expensive and the defender would rather incur the cost of being in a poor system state (see Equation (\ref{eq:statecost})).

An interesting (related) observation can be seen by analyzing the characteristics of the observer states and how these characteristics influence \emph{when} the policy undergoes a switch as $r$ increases. 
%To see why this behavior is intuitive, it is useful to consider the following example.
Consider Figure \ref{fig:reimage}(a), and observe the behavior of the optimal policy around the re-image cost of $r = 20$. There is a collection of observer states (with indices 74 -- 87) that contain the $(F,F)$ element (both computers are in the remote compromised state)  %, that is, both computers are fully communicable -- the most costly operating mode for the defender. 
where the optimal policy specifies a switch from re-image to null. In these observer states, the defender believes that the true system state is so poor that, even if the a computer were to be re-imaged, the events in $\A$ would cause the system to transition back to a poor state in so few iterations that the defender would just be wasting its resources by re-imaging. That is, the number of time steps that it takes for the system to return to a poor state is not high enough to justify the cost that the defender must incur to keep the system in a secure operating mode. For this reason, in these observer states, the defender exhibits the passive behavior of \emph{giving up} by choosing the cheapest action, $N^d$. An interesting related observation is that for other observer states in the system (the observer states that do not contain the element $(F,F)$) the optimal policy specifies a switch away from re-image at a higher re-image cost (around $r\in [25\,\,\, 26]$). In these observer states the defender views the process of securing the system as economically efficient because it can be returned to a secure operating mode in a small enough number of iterations (compared to the observer states that contain the system state $(F,F)$). This observed behavior reflects the fact that attacks are progressive and that time has value in our model.

Another observation is that there are sets of parameters for which the sense action is useful (as seen starting in Figure \ref{fig:reimage}(a) around $r=2$ and peaking in Figure \ref{fig:reimage}(b) around $r=25$). 
% connect to why the sense action results in cheaper V(S') than for re-image,
% V(S') cheaper in sense compared to null
In these cases the act of sensing a computer results in a split observer state that has a lower future cost than if the defender were to choose either null or re-image.
% what does it mean for an observer state to have a lower future cost
% mention something about for moderate values of $r$ we actually have more sensing actions than for low values of $r$
Thus, paying the cost to sense can result in the defender having a better idea of the underlying system state and thus make a wiser decision on which future action to take. However, for low values of $r$, we can see that the defender prefers to re-image over obtaining a better estimate of the system (and similarly for high values of $r$, the defender prefers to take the null action). This behavior highlights the duality between estimation and control.

Interestingly, sensing remains an optimal action even for high values of $r$ when there is no re-image action prescribed in the optimal defense policy. In these cases, even though sensing does not change the state of the network, it refines the defender's information which then results in a lower future cost for the defender. Even though the sense action is more expensive than the null action, this lower future cost causes the defender to choose sense over null.

The intent of determining an optimal policy is to offer a set of procedures for the defender such that the network is able to be kept as secure as possible. After the defender specifies its costs for actions and costs for states, the optimal policy specifies a procedure that the defender should follow. For each action the defender takes, $d\in\D$, and for each event it observes, $a'\in\A'$, the resulting observer state is known through the dynamics of the observer state. For each of these observer states resulting from the sequence of defender actions and observed events, the optimal policy specifies whether to sense or re-image a particular computer, or to wait and do nothing. The resulting defender behavior will keep the network as secure as possible under the min-max cost criterion.

\section{Conclusion and Reflections}
\label{sec:conclusion}

In this paper we have proposed a supervisory control approach to dynamic cyber-security. We have taken the viewpoint of the defender whose task is to defend a network of computers against progressive attacks. Some of the attacker actions are unobservable by the defender, thus the defender does not have perfect knowledge of the true system state. We define an observer state for the defender to capture this lack of perfect knowledge.

We have assumed that the defender takes a conservative approach to preserving the security of the system. We have used the min-max performance criterion to capture the defender's conservative approach. 

Dynamic programming was used to obtain an optimal defender policy to Problem (\ref{eq:defprobprime}). The numerical results show that the optimal policy exhibits a threshold behavior when the cost of actions are varied. We have also observed the duality of estimation and control in our optimal policy.

We believe that our approach is suitable for modeling interactions between an attacker and a defender in general security settings. In general, we can use our approach to study dynamic defense against attacks in a network of $N$ resources each with $\M$ (orderable) security levels and $\M-1$ security boundaries. The attack actions can penetrate through some of these boundaries to compromise a resource, or use a compromised resource to attack other resources in the network. Some of these actions can be unobservable to the defender. On the other hand, the defender can take actions to change the state of resources to a more secure operating mode or sense the system state to obtain more refined information about the system's status. 

The model we have defined is rich enough to be extended to capture more complicated environments. Some examples of such environments can be heterogeneity of the network's computers\footnote{Placing an \emph{importance} weight on each computer.} or the introduction of a dummy computer\footnote{The dummy computer contains no sensitive information and is meant to mislead the attacker.} into the system so as to increase the network's resiliency to attacks. %The numerical results and some structural results presented in this paper can potentially be extracted for these systems. 

One bottleneck of our approach is that the number of states and transitions grows exponentially with the number of computers. One solution to this is to use a hierarchical decomposition for the system. For example an \emph{Internet Service Provider} (ISP) can model a collection of nodes in their network as one region (resource). Once a non-secure region is observed in the system, the ISP can more carefully analyze the nodes within that region and take appropriate actions. Approximate dynamic programming methods could also be useful in dealing with systems with a large number of computers.\\

\noindent{\bf Acknowledgement}\\

This work was supported in part by NSF grant CNS-1238962 and ARO MURI grant W911NF-13-1-0421. The authors are grateful to Eric Dallal for helpful discussions.\\

\noindent {\large{\bf A\quad Appendix -- UMDES-LIB}}\\

The UMDES-LIB library \cite{umdes1} is a collection of C-routines that was built to study discrete event systems that are modeled by finite state automata. Through specification of the states and events of a system automaton (along with the controllability and observability of events), the library can construct an entity termed the \emph{observer automaton.} In our problem the observer automaton is the defender's observer automaton, since we take the viewpoint of the defender. Thus, the observer automaton consists of the defender's observer states.
% and intermediate-defender

In this appendix we describe an automated process\footnote{Source code is available upon request.} for extracting the defender's observer state from the system automaton that makes use of UMDES-LIB. This requires first \emph{constructing the system automaton} in an acceptable format for the library while preserving all the features of our model. After running the library on the provided system automaton, we \emph{extract the defender's observer state} from the observer automaton output. 
% but I thought the observer automaton already consisted of the defender's information state (I suppose there are 'extra' information states in the observer automaton)
This method allows one to construct the defender's observer state for any number of computers.\footnote{The only bottleneck being the (potentially large) dimensionality of the problem.}\\

\noindent {\bf Constructing the System Automaton.} The input that we provide to UMDES-LIB is the system automaton from the defender's viewpoint, as illustrated earlier in Figure \ref{fig:bipartite}. 

In order to preserve all features of our model in the resulting observer automaton, we need to introduce additional sensing actions. Recall that the sense action, $\{E_i\}_{i\in\N}$, causes the system automaton to transition to the same state as the null action, $N^d$ (see Figure \ref{fig:bipartite}). However, as stated in Section \ref{sec:model}, the sense action updates the information state of the defender. In order to ensure that UMDES-LIB captures this functionality, we expand the sense action $E_i$ for each computer $i$ into $|\M|$ distinct actions, denoted by $E_i^{z^i}$, which represent sensing computer $i$ when it is in state $z^i\in\M$. This results in a reduced level of uncertainty for the defender as it splits the observer state into, at most, $|\M|$ possible sets of observer states. The admissible actions from $\{E_i^{z^i}\}_{z^i\in\M}$, at a given system state, are the sense actions that correspond to the true system state. For example, from the system state $Z_t = (N,R,W)$, the admissible sense actions are $E_1^{N}$, $E_2^{R}$, and $E_3^{W}$. The above example of the expanded sense action is perhaps worrisome at first glance -- if the only admissible sense actions from the current state are the ones that correspond to the current state of the computer, then the defender will know what the current state of each computer is, eliminating the need for a sense action. However, the observer state that is obtained from each expanded sense action is the same as the observer state that is obtained if the defender were to observe the true, unknown state of a computer. 

% a priori vs. a posteriori information
Running UMDES-LIB on the system automaton with the expanded sense actions results in the observer automaton.\\

{\bf Extracting the Defender's Observer State.} The output of UMDES-LIB is the observer automaton, from which we must extract the defender's observer state. 
% add introductory sentence
First, since the defender does not have the ability to choose the expanded sense actions, $E_i^{z^i}$, we re-group them into a single, non-deterministic action, $E_i\in\D$, for each $i\in\N$. 
Next, we need to extract the function, $f: \Ss\times \D \times \A' \rightarrow \Ss$ from the observer automaton. The observer automaton, generated by UMDES-LIB, takes the form of a bipartite graph; one collection of states of the bipartite graph is observer states over system states $\Z$, denoted $\Ss$, whereas the other collection is observer states over intermediate states $\tilde \Z$, denoted $\tilde \Ss$. %The (controllable) actions in $\D$ are the events from the observer states over $\Z$ to the observer states over $\tilde \Z$. 
Defender actions, $\D$, are the only admissible actions from observer states $\Ss$. The defense action $d\in\D$ causes a transition\footnote{This transition may be non-deterministic due to the sense action.} to an observer state in $\tilde \Ss$, where only events in $\A'$ are admissible. Each event $a'\in\A'$ causes a transition back to an observer state in $\Ss$. Repeating this process for all observer states in $\Ss$, actions $d\in\D$, and events $a'\in\A'$, the function $f: \Ss\times \D \times \A' \rightarrow \Ss$ is defined. 
%This is done by starting at an information state $S\in\Ss$, finding the information states in $\tilde\Ss$ that are reachable by taking the action $d\in\D$, then tracing the transition back to information states in $\Ss$ due to admissible actions, $a'\in\A'$, at the respective information states in $\tilde\Ss$. The event corresponding to this round-trip transition (back to $\Ss$) is (Mohammad: From here)replaced by the bundle of the controllable events in $\D$ which caused the initial transition to $\tilde\Ss$ and the uncontrollable event $a'\in \A'$ which takes the system from $\Ss$ to $\tilde\Ss$ (Mohammad: To here). Repeating this for all $S\in\Ss$, $d\in\D$, and $a'\in\A'$ results in the mapping $f$ and thus defines the defender's information state.
To construct the set $\Q(S, d, Z)$ we follow the approach described in Section 4.1 and illustrated by Example 1.

\bibliographystyle{abbrv}
\bibliography{references}
\end{document}